\documentclass[11pt]{article}

\usepackage{amsfonts}
\usepackage{amsmath}
\usepackage{amssymb}
\usepackage{amsthm}
\usepackage[mathscr]{eucal}
\usepackage{bbold}
\usepackage{braket}
\usepackage{color}
\usepackage{dsfont}
\usepackage{enumitem}
\usepackage{framed}
\usepackage{jheppub}
\usepackage{mathtools}
\usepackage{physics}
\usepackage[normalem]{ulem}
\usepackage{tensor}
\usepackage{thmtools}
\usepackage{thm-restate}
\usepackage{ulem}

\definecolor{AFcolor}{rgb}{0.1,0.7,0.1}

\definecolor{SHCcolor}{rgb}{0.9,0.1,0.1}

\newcommand{\ignore}[1]{}

\DeclareMathOperator{\interior}{int}
\DeclareMathOperator{\Ric}{Ric}

\makeatletter\def\@fpheader{~}\makeatother
\title{Conformal Rigidity from Focusing}

\preprint{MIT-CTP/5307}

\author[a]{{\AA}smund Folkestad}
\emailAdd{afolkest@mit.edu}
 
\author[b]{and Sergio Hern\'andez-Cuenca}
\emailAdd{sergiohc@ucsb.edu}

\affiliation[a]{Center for Theoretical Physics, Massachusetts Institute of Technology, Cambridge, MA 02139, USA}
\affiliation[b]{Department of Physics, University of California, Santa Barbara, CA 93106, USA}

\abstract{
The null curvature condition (NCC) is the requirement that the Ricci curvature of a Lorentzian manifold be nonnegative along null directions, which ensures the focusing of null geodesic congruences.
In this note, we show that the NCC together with the causal structure significantly constrain the metric.
In particular, we prove that any conformal rescaling of a spacetime with zero null curvature introduces either geodesic incompleteness or negative null curvature, provided the conformal factor is non-constant on at least one complete null geodesic. 
In the context of bulk reconstruction in AdS/CFT, our results combined with the technique of light-cone cuts can be used in spacetimes with zero null curvature to reconstruct the full metric in regions probed by complete null geodesics reaching the boundary.
For spacetimes with null curvature, our results constrain the conformal factor, giving an approximate reconstruction of the metric.
}
 
\begin{document}
 

\maketitle


\section{Introduction}\label{sec:intro}
Energy conditions are indispensable tools in the study of gravity. 
They enable the proof of important results which are true in theories describing nature, but not in Lorentzian
geometry in general. Among the most important is the null energy condition
(NEC), which underpins the famous area theorem by Hawking \cite{Hawking:1971tu}, the Penrose singularity theorem
\cite{Penrose:1964wq}, in addition to results related to holography
\cite{Gao:2000ga,Wall:2012uf, Headrick:2014cta,Bousso:2015mqa}, wormholes \cite{Morris:1988tu}, cosmic censorship
\cite{Chen:2020hjm, EngFol20}, topological censorship \cite{Friedman:1993ty}, and causality
\cite{Tipler:1976bi,Hawking:1991nk,Olum:1998mu,Visser:1998ua}.\footnote{See also e.g. \cite{Ford:1994bj,Graham:2007va,Hartman:2016lgu,Bousso:2015wca,Balakrishnan:2017bjg, Freivogel:2018gxj, Leichenauer:2018tnq} for generalizations of the NEC to weaker conditions more relevant to quantum fields.} The great utility of the NEC comes from that fact that it ensures the focusing of null geodesics, a physically motivated condition capturing the attractive nature of gravity.

In the Hawking area theorem, an assumption on the causal structure of spacetime together with the NEC is sufficient
to derive how areas of certain surfaces behave -- that is, properties of the underlying geometry. While the causal structure is invariant under Weyl transformations, areas are not. Thus, when causality is combined with the focusing of null geodesic (as implied by the NEC), we gain control over the geometry. Other remarkable examples of the interplay between causal structure and focusing in constraining geometry were found in \cite{May:2019odp, May:2021nrl, May:2021zyu, May:2021raz}. 
In these works, area comparison theorems for extremal surfaces in spacetimes with Anti-de Sitter (AdS) asymptotics
were derived from only causal relations and the null curvature condition (NCC), 
a property equivalent to the NEC when assuming the Einstein equation.
Together, these results hint at a pattern where
causal structure and energy conditions together constrain geometry to a substantial degree and in physically relevant ways.

In this note we further elaborate on this pattern, and exhibit situations where the NCC together with the
conformal structure is enough fix the geometry in various ways, sometimes even uniquely.
More specifically, as one of our results we will show that if a conformal class contains a geodesically complete spacetime with vanishing null curvature, then every other
geometry in its conformal class either violates the NCC, or suffers from geodesic incompleteness along every geodesic along which the relative conformal factor is non-constant. In other words, in these situations, causal structure, the NCC and geodesic
completeness are enough to uniquely fix the metric.

We begin in Section \ref{sec:basics} by setting up the notation and introducing basic ideas which will be used throughout. Sections \ref{sec:rigidity} is devoted to presenting our main results in a self-contained manner, with their proofs relegated to Section \ref{sec:proofs} for ease of reading. Finally, in Section \ref{sec:discuss} we discuss our findings and point out applications of our results to AdS/CFT, particularly for bulk reconstruction.

\section{Weyl Transformations and Null Curvature}\label{sec:basics}
Throughout this note, a spacetime is defined as a pair $(M,g)$ consisting of a $C^2$ Lorentzian manifold $M$ and a $C^2$ metric tensor $g : \Gamma(TM)\times\Gamma(TM)\to\mathbb{R}$, where $\Gamma(TM)$ denotes vector fields of the tangent bundle $TM$ of $M$. The spacetime dimension is denoted by $D\equiv\dim M$ and we assume $D\ge3$ throughout. The causal structure of a spacetime $(M, g)$ is invariant under  \textit{Weyl transformations}, which are rescalings of the metric tensor $g \rightarrow \Omega^2 g$ by some positive function $\Omega$ on $M$. Hence, the specification of the causal structure of a spacetime is equivalent to specifying the conformal class of its metric, which leaves an overall factor in the metric undetermined.

One may deem a Weyl transformation pathological if its \textit{Weyl factor} $\Omega$ is not $C^2$ or exhibits an undesirable asymptotic behavior. In particular, denoting the conformal boundary of $M$ by $\partial M$,  one could in principle have $\Omega(p)\to0$ or $\Omega(p)\to\infty$ as $p$ is taken towards $\partial M$ along some sequence of points. In such cases, the extendibility and completeness properties of the spacetime or the topology of its conformal boundary may be altered. It is thus useful to distinguish these more ``violent'' Weyl transformations from those which do not alter such properties. To do so, it will be helpful to parameterize Weyl transformations using $\Omega = e^{\omega}$, with $\omega : M \to \mathbb{R}$ sometimes referred to as the \textit{Weyl exponent}. In particular, it is easy to see that $\Omega$ will not alter completeness properties of the underlying spacetime if $\omega$ is bounded (see Lemma~\ref{lem:listcomplete} for a more precise statement).

Consider now a null vector field $k\in\Gamma(TM)$, $k^2=0$, and its associated null curvature
\begin{equation}
    R_{kk} = \Ric(k, k),
\end{equation}
where $\Ric$ is the Ricci curvature of $(M, g)$, also written as $R_{ab}$ in abstract index notation. When $R_{kk}\geq 0$ for every null vector field $k$ we say that the NCC holds for $(M, g)$ and, if $R_{kk}=0$ for every such $k$, then we say that $(M,g)$ has zero null curvature. Importantly, the condition of zero null curvature does not trivialize the Ricci tensor of the spacetime -- in fact, e.g. any region of spacetime with $\Ric = f g$ for any scalar $f$ will satisfy this property.\footnote{Notice this is more general than the condition that the manifold be Einstein.}

The NCC is a well-defined geometric condition on Lorentzian manifolds independent of a theory of gravity (be it classical or semiclassical) or any physics at all, even if it is only physically motivated in the classical limit. The most standard forms of field theory matter coupled to Einstein gravity, such as minimally coupled scalars, gauge fields, and $p$-forms (all with canonical kinetic terms), all source classical spacetimes respecting the NCC. The same is true for the bosonic sector of type IIA, IIB and $D=11$ supergravity,\footnote{This is immediate for $D=11$ supergravity, whose bosonic sector only contains a free $3$-form (besides Chern-Simons terms of no relevance here). For type IIA and IIB supergravity, the bosonic matter stress tensors are essentially those of free bosons and $p$-forms with canonical kinetic parts, up to overall dilatonic prefactors having no effect on the NEC \cite{Hamilton:2016ito}. Various Kaluza-Klein dimensional reductions can also be checked to preserve the NEC \cite{Parikh:2014mja} -- e.g. for $\mathcal{N}=2$ supergravity in $D=4$ one just gets Einstein-Maxwell. We thank Gary Horowitz for helpful discussion on this.} and in classical bosonic string theory the NCC can be derived to leading order in $\alpha'$ \cite{Parikh:2014mja}. In the context of AdS/CFT, it is furthermore known that classical matter violating the NCC can lead to various pathologies in the holographic dictionary \cite{Headrick:2014cta}, and so assuming the NCC in the classical limit is well motivated there. The assumption of the NCC is equivalent to the condition that gravity always focuses null geodesic congruences.

Knowledge of the conformal class of a spacetime is clearly not enough to fix its metric, but what if we also demand the NCC? Given a spacetime that obeys the NCC, one conformally equivalent to it need not do so. Under a Weyl transformation with Weyl exponent $\omega$, the Ricci tensor becomes
\begin{equation}\label{eq:ricciweyl}
   R^{(\omega)}_{ab} = R_{ab} - (D-2) \left( \nabla_a \nabla_b \omega - \nabla_a \omega \nabla_b \omega \right) - \left( \nabla^2 \omega + (D-2) (\dd\omega)^2 \right) g_{ab}.
\end{equation}
The transformed null curvature then reads
\begin{equation}\label{eq:expnccterm}
    R_{kk}^{(\omega)} = R_{kk} - (D-2) k^a k^b \left( \nabla_a \nabla_b \omega - \nabla_a \omega \nabla_b \omega \right).
\end{equation}
It follows that $\omega$ defines an NCC-preserving Weyl transformation on $(M,g)$ if and only if, everywhere and for any $k$,
\begin{equation}\label{eq:expncc}
    R_{kk}^{(\omega)} \geq 0 \quad\Longleftrightarrow\quad (D-2)k^a k^b \left( \nabla_a \nabla_b \omega - \nabla_a \omega \nabla_b \omega \right) \leq R_{kk}.
\end{equation}
Consider now a null geodesic curve $\gamma : \mathbb{R}\supseteq I \to M$ that is affinely parameterized by $\lambda\in I$. For its tangent vector field $k^a\equiv(\partial_\lambda)^a$, which obeys the geodesic equation $\nabla_k k = 0$, the condition on the right-hand side of \eqref{eq:expnccterm} becomes simply
\begin{equation}\label{eq:secschw}
    \nabla_k^2 \omega - (\nabla_k \omega)^2 =  \frac{ \dd^2 \omega_\gamma }{ \dd \lambda^2} -\left( \frac{ \dd \omega_\gamma }{ \dd \lambda }\right)^2\leq \frac{1}{D-2}R_{kk},
\end{equation}
where $\omega_\gamma \equiv \omega \circ \gamma$ is just the Weyl exponent along the pertinent geodesic. This inequality is the starting point for our work.

\section{Conformal Rigidity Results}
\label{sec:rigidity}

In this section we derive the main inequalities from which our theorems follow. We also state the theorems, with proofs postponed to the next section. In what follows, let $k^a\equiv(\partial_\lambda)^a$ be tangent to an affine null geodesic $\gamma$ with affine parameter $\lambda\in I\subseteq \mathbb{R}$. When referring to any scalar function on $M$, we will implicitly work with its composition with the curve $\gamma$, as was done above e.g. in \eqref{eq:secschw}.\footnote{In other words, to avoid cluttering the notation, by e.g. $\omega$ we will often mean $\omega_\gamma$ -- any exception to this should be clear from context.}

Substituting $\omega = -\ln u$ for a strictly positive function $u$ into \eqref{eq:expnccterm}, we find that to satisfy the NCC in the Weyl-rescaled spacetime we need
\begin{equation}
    R_{kk}^{(\omega)}(\lambda) = R_{kk}(\lambda) + (D-2)\frac{ u''(\lambda) }{ u(\lambda) } \geq 0.
\end{equation}
Multiplying through by $u$ and integrating once and twice yields the following inequalities:
\begin{align}
    u'(\lambda) &\gtreqless u'(\lambda_0)  - \frac{ 1 }{ D-2 } \int_{\lambda_0}^{\lambda}\dd \lambda'  \;
    R_{kk}(\lambda') u(\lambda'), &  \forall\lambda &\gtreqless \lambda_0, \label{eq:udifbound1}\\
    u(\lambda) &\geq u(\lambda_0) + u'(\lambda_0) (\lambda - \lambda_0) - \frac{ 1 }{ D-2 }
    \int_{\lambda_0}^{\lambda}\dd \lambda' \; (\lambda-\lambda') R_{kk}(\lambda') u(\lambda'), & \forall
    \lambda &\in I, \label{eq:ubound}
\end{align}
for any $\lambda_0\in I$. 
As a special case of these inequalities, consider a situation where $I=(-\infty, c)$, and assume that the following limits exist
\begin{equation}
    \lim_{\lambda \rightarrow -\infty}u(\lambda) \neq 0, \qquad \lim_{\lambda\rightarrow -\infty} u'(\lambda).
\end{equation}
Then we must have $\lim_{\lambda \rightarrow -\infty}(\lambda u'(\lambda)) =0$, for otherwise we would need $u'(\lambda) \sim \mathcal{O}(1/\lambda)$ at large $\lambda$, giving that $u$ diverges logarithmically. Observing that $u=1/\Omega$ and taking $\lambda_0=-\infty$, we then get
\begin{align}
    -\frac{ \dd }{ \dd \lambda } \left( \frac{1}{\Omega(\lambda)} \right) &\leq \frac{1}{ D-2 }\int_{-\infty}^{\lambda} \dd \lambda' \;\frac{R_{kk}(\lambda')}{\Omega(\lambda')}, \label{eq:omegadifbound}\\
    \frac{ 1 }{ \Omega(\lambda) } &\geq \frac{ 1 }{ \Omega(-\infty) } - \frac{ 1 }{ D-2
    }\int_{-\infty}^{\lambda} \dd \lambda' \;(\lambda - \lambda') \frac{  R_{kk}(\lambda') }{
\Omega(\lambda') }. \label{eq:omegabound}
\end{align}
Thus, in the situation where $\Omega$ approaches a non-zero constant $\Omega(-\infty)$ as $\lambda\to-\infty$ so that the limit $\lim_{\lambda \rightarrow -\infty} \Omega'(\lambda)$ exists, we encounter two possibilities:
\begin{itemize}
    \item If the original spacetime has zero null curvature $R_{kk}=0$ and we want to preserve the NCC, then we see from \eqref{eq:omegabound} that $\Omega$ must be monotonically non-increasing along $\gamma$ as we move away from the direction in which $\gamma$ is complete.
    \item If the original spacetime satisfies the NCC with nonvanishing null curvature $R_{kk}\geq0$, then the potential increase of $\Omega(\lambda)$ along a complete null geodesic $\gamma$ is bounded from above by the previous null curvature the geodesic has encountered, as dictated by \eqref{eq:omegabound}.
    In particular, if $\Omega$ satisfies $\Omega(\lambda) \geq \Omega_{\rm min} > 0$ along some complete null geodesic $\gamma$, then 
    \begin{equation}\label{eq:ancc}
        - \frac{ \dd }{ \dd \lambda } \left(\frac{ 1 }{ \Omega(\lambda) }\right) \leq \frac{ \mathcal{R}_{\gamma} }{ (D-2)\Omega_{\rm min}
        }, \qquad \mathcal{R}_{\gamma} \equiv \int_{-\infty}^{\infty}\dd \lambda' R_{kk}(\lambda'),
    \end{equation}
    where we have defined the \emph{averaged null curvature} $\mathcal{R}_{\gamma}$ along the null geodesic $\gamma$.\footnote{This object and the condition that it be nonnegative have appeared previously in e.g. \cite{Borde:1987qr,Engelhardt:2016aoo}. In general relativity, it is analogous to the averaged null energy, arising in discussions of the averaged NEC \cite{Ford:1994bj,Graham:2007va,Hartman:2016lgu}.}
    Integrating now \eqref{eq:ancc} from $\lambda_0$ to $\lambda$, we get
    \begin{equation}\label{eq:ancc2}
        \frac{ 1 }{ \Omega(\lambda) } \gtreqless \frac{ 1 }{ \Omega(\lambda_0) } - \frac{ \mathcal{R}_{\gamma} }{ (D-2)\Omega_{\rm min}
        }(\lambda-\lambda_0), \qquad \lambda \gtreqless \lambda_0.
    \end{equation}
\end{itemize}

When $R_{kk}=0$ and we do not assume anything about the asymptotic behavior of $\Omega$, we find the following more general statement, from which most of our results will follow: 
\begin{restatable}{nlemma}{lemcases}
    \label{lem:cases}
   Let $(M,g)$ be a spacetime with zero null curvature. Let $\Omega \in C^2(M)$ be a positive function, and let $\gamma$ be any complete null geodesic affinely parameterized by $\lambda \in \mathbb{R}$. Then at least one of the following is true:
    \begin{enumerate}
        \item\label{c1} The spacetime $(M, \Omega^2 g)$ violates the NCC along $\gamma$.
        \item\label{c2} The Weyl factor $\Omega$ is constant along $\gamma$.
        \item\label{c3} $\lim\limits_{\lambda \rightarrow \infty } \Omega(\lambda) = 0$ or $\lim\limits_{\lambda \rightarrow -\infty } \Omega(\lambda) = 0$, with $\Omega$ decaying at least as fast as $\mathcal{O}(1/|\lambda|)$.
    \end{enumerate}
\end{restatable}

With this lemma in hand, together with another technical lemma relegated to the next section, we find the following series of results.
\begin{restatable}{nthm}{geocomplete}
    \label{thm:geomplete}
    Let $(M,g)$ be a connected and geodesically complete spacetime with zero null curvature.
    If $\omega\in C^2(M)$ is bounded, then either $\omega$ is constant on $M$ or $(M, e^{2\omega} g)$ violates the NCC.
\end{restatable}

\begin{restatable}{nthm}{thmcompVSncc}
    \label{thm:compVSncc}
    Let $(M,g)$ be a spacetime with zero null curvature. If $\omega\in C^2(M)$ is bounded and non-constant, then either $(M,e^{2\omega} g)$ is null-geodesically incomplete, or it violates the NCC. If $(M, g)$ is null-geodesically complete, then the previous is true also when $\omega$ is only bounded from below.
\end{restatable}
\noindent The first theorem states that a conformal rescaling of a geodesically complete spacetime with zero null curvature with bounded $\omega$ always causes NCC violation. The second states that in a general spacetime with zero null curvature, after a bounded conformal rescaling, negative null curvature or geodesic incompleteness must always be present. In the next theorem, we show that an NCC-preserving conformal rescaling with bounded $\omega$ is only ever possible in incomplete spacetimes:

\begin{restatable}{nthm}{thmincomplete}
    Let $(M,g)$ be a spacetime with zero null curvature. Let $\omega \in C^2(M)$ be bounded and non-constant. If $(M, e^{2\omega} g)$ preserves the NCC, then $(M, g)$ is null-geodesically incomplete.
\end{restatable}

Next we turn to some applications for spacetimes with commonly used asymptotics.
\begin{restatable}{nthm}{thmads}\label{thm:thmads}
    Let $(M, g)$ be an asymptotically locally AdS (AlAdS) spacetime with zero null curvature. Let $\Omega \in C^2(M)$ be a positive function and $A\subseteq M$ the set of points $p \in M$ 
    lying on at least one complete null geodesic with both endpoints on $\partial M$. 
    Assume that $\Omega$ extends to $\partial M$ with the same value everywhere.
    Then at least one of the following is true:
    \begin{enumerate}
        \item The spacetime $(M, \Omega^2 g)$ violates the NCC on $A$.
        \item The Weyl factor $\Omega$ is constant on $A$. 
        \item $(M, \Omega^2 g)$ is not AlAdS.
    \end{enumerate}
\end{restatable}

\begin{restatable}{nthm}{thmAF}
\label{thm:thmaf}
Let $(M, g)$ be an asymptotically flat spacetime with zero null curvature. Let $\Omega \in C^2(M)$ be a positive function and $A$ the set of points $p \in M$ 
    lying on at least one complete null geodesic with both endpoints lying on $\mathscr{I}^+ \cup \mathscr{I}^-$.
    Assume that $\Omega$ extends to $\mathscr{I}^+ \cup \mathscr{I}^-$ with the same value everywhere.
    Then at least one of the following is true:
    \begin{enumerate}
        \item The spacetime $(M, \Omega^2 g)$ violates the NCC on $A$.
        \item The Weyl factor $\Omega$ is constant on $A$. 
        \item $(M, \Omega^2 g)$ is not asymptotically flat.
    \end{enumerate}
\end{restatable}
\noindent These two analogous results prove strong versions of conformal rigidity for spacetimes asymptoting to AdS or flat metrics when subject to the NCC. In particular, they show that with prescribed asymptotics, regions probed by boundary-anchored null geodesics do not admit any nontrivial NCC-preserving Weyl transformation where $\Omega$ asymptotes to a constant.

Finally, as an example of how causal structure and the NCC can be used to constrain spacetime also in regions without boundary-anchored null geodesics, we now investigate the BTZ black hole \cite{BTZ}, where all null geodesics except the closed photon orbits are incomplete \cite{Cruz:1994ir}. We find the following:
\begin{restatable}{nthm}{thmbtz}
\label{thm:btz}
    Let $(M, g)$ be the BTZ black hole spacetime and $r$ the area radius. Then $(M, e^{2\omega} g)$ violates the NCC unless $\omega$ preserves spherical symmetry outside the horizon. Furthermore, if the limits $\lim_{r\rightarrow \infty} \omega(r)$  and $\lim_{r \rightarrow \infty} \omega'(r)$ exist, then outside the horizon
    \begin{equation}
    \omega'(r) \geq 0.
    \end{equation}
\end{restatable}

\section{Proofs}
\label{sec:proofs}

\lemcases*
\begin{proof}

Assume $(M, \Omega^2 g)$ satisfies the NCC and $\Omega$ is not everywhere constant along $\gamma$.
Then there exists an affine parameter $\lambda_0$ such that $u'(\lambda_0) \neq 0$. 
From \eqref{eq:ubound}, one has
\begin{equation}
    u(\lambda) \geq u'(\lambda_0)(\lambda-\lambda_0) + u_0.
\end{equation}
Depending on the sign of $u'(\lambda_0)$, we get that $u$ diverges at least as fast as $|\lambda|$ as
either $\lambda \rightarrow \infty$ or $\lambda \rightarrow -\infty$. Since $\Omega=1/u$, the result follows.

\end{proof}

\begin{restatable}{nlemma}{lemlistcomplete}
    \label{lem:listcomplete}
    Let $\gamma$ be an inextendible null geodesic on a spacetime $(M,g)$, and $\Omega \in C^2(M)$ a positive function.
    \begin{itemize}
        \item If $1/\Omega_{\gamma}$ is bounded above and $\gamma$ is complete in $(M, g)$, then $\gamma$ is complete in $(M, \Omega^2 g)$.
        \item If $\Omega_{\gamma}$ is bounded above and $\gamma$ is complete in $(M, \Omega^2 g)$, then $\gamma$ is complete in $(M, g)$.
    \end{itemize}
\end{restatable}

\begin{proof}
    Note first that $\gamma$ must be inextendible as a geodesic in $(M, \Omega^2 g)$, 
    since $(M, \Omega^2 g)$ and $(M, g)$ are identical at the manifold level, so an
    endpoint in $(M, g)$ is an endpoint in $(M, \Omega^2 g)$.
    The affine parameters $\lambda$ and $\tilde{\lambda}$ of $\gamma$ respectively in $(M, g)$ and $(M, \Omega^2 g)$ are related by
    \begin{equation}
    \label{eq:tildel}
        \frac{ \dd \tilde{\lambda}}{ \dd\lambda } = c\Omega_\gamma^2,
    \end{equation}
    where one can fix $c=1$ and $\lambda(\tilde{\lambda}=0)=0$ without loss of generality. 
    Assume $1/\Omega_\gamma$ is bounded above so there exists an $\epsilon > 0$ such that $\Omega_\gamma\ge\epsilon$.
    If $\gamma$ is complete in $(M, g)$, then 
    \begin{equation}
        \lim_{\lambda\to\infty} \tilde{\lambda} = \int_0^{\lambda} \dd \lambda \; \Omega_\gamma^2 \geq \epsilon^2 \int_0^{\infty} \dd \lambda = \infty,
    \end{equation}
    so $\gamma$ is complete in $(M, \Omega^2 g)$. The proof for $\Omega_\gamma$ bounded is analogous, replacing $\Omega$ by $1/\Omega$. 
\end{proof}

\geocomplete*
\begin{proof}
    Assume $(M, e^{2\omega} g)$ satisfies the NCC. Since $(M,g)$ is complete and $\omega$ is bounded by assumption, Lemma \ref{lem:cases} implies $\omega$ is constant along every
    null geodesic. This will remain true along any broken null geodesic (i.e., any piecewise differentiable curve that is the concatenation of null geodesic segments). This is because all segments are extendible to complete geodesics, with concatenated intersections in pairs. Now, because the manifold is connected, there exists a path $\gamma$ between any two points on $M$. One can approximate any such path in an arbitrarily small neighborhood of $\gamma$ by a broken null geodesic also connecting the two endpoints. Hence $\omega$ is the same at any two points of $M$, proving the desired result.
\end{proof} 

\thmcompVSncc*
\begin{proof}
    Assume $(M, \Omega^2 g)$ with $\Omega = e^{\omega}$ does not violate the NCC. If $\Omega$ is not constant, there 
    must exist a null vector field $k$ and a point $p$ such that $\nabla_k \Omega|_{p} \neq
    0$. Let now $\gamma$ be an inextendible null geodesic passing through $p$ with tangent $k$ there affinely parameterized by $\lambda$. 
    If this geodesic is incomplete in $(M, g)$ we are done: since $\Omega$ is bounded,
    Lemma~\ref{lem:listcomplete} would also imply incompleteness in $(M, \Omega^2 g)$.
    So assume $\gamma$ is complete in $(M, g)$.  

    By Lemma~\ref{lem:cases} we know that $\Omega \sim 1/\abs{\lambda}$ as either $\lambda \rightarrow
    -\infty$ or $\lambda \rightarrow \infty$. Assume without loss of generality the latter case. 
    The affine parameter $\tilde{\lambda}$ of $\gamma$ in $(M, \Omega^2 g)$ is
    related to $\lambda$ by \eqref{eq:tildel}
    with the same conventions. 
    Since there must exist constants $C$ and $\epsilon>0$ such that
    \begin{equation}
        \Omega\left(\lambda\right) < \frac{ C }{ \lambda }, \qquad \forall \lambda > \frac{ 1 }{ \epsilon },
    \end{equation}
    one can write
    \begin{equation}
            \lim_{\lambda\to\infty} \tilde{\lambda} = \int_{0}^{\infty} \dd\lambda \; \Omega(\lambda)^2 = 
        \int_{0}^{1/\epsilon} \dd\lambda \; \Omega(\lambda)^2 + \int_{1/\epsilon}^{\infty} \dd\lambda \; \Omega(\lambda)^2.
    \end{equation}
    But the two terms on the right-hand side are finite, since $\Omega\in C^2(M)$ and
    \begin{equation}
        \int_{1/\epsilon}^{\infty} \dd\lambda \; \Omega(\lambda)^2 < C^2 \int_{1/\epsilon}^{\infty}
        \frac{ \dd\lambda }{ \lambda^2 } = C^2 \epsilon.
    \end{equation}
    Hence $\gamma$ is incomplete in $(M, \Omega^2 g)$.
    
    Finally, note that the boundedness of $\Omega$ was only used to deduce that $\gamma$ was complete in $(M, g)$. If $(M, g)$ is geodesically complete, the assumption of $\Omega$ being bounded can be removed.
\end{proof}

\thmincomplete*
\begin{proof}
Boundedness of $\omega$ implies boundedness of $\Omega$ and $1/\Omega$.
By boundedness of $\Omega$ and Theorem~\ref{thm:compVSncc}, $(M, \Omega^2 g)$ is null-geodesically incomplete. 
By boundedness of $1/\Omega$ and Lemma~\ref{lem:listcomplete}, this means $(M, g)$ is geodesically incomplete. 
\end{proof}

\thmads*
\begin{proof}
Assume $(M, \Omega^2 g)$ obeys the NCC. By definition of AlAdS \cite{Fischetti:2012rd} there exists a function $z$ on $M$ such that $z^2 g$ extends smoothly to a metric on the conformal boundary $\partial M$, with
$z|_{\partial M} = 0$, $\dd z|_{\partial M} \neq 0$ and $z|_{\interior M}>0$.
By Lemma \ref{lem:cases}, on each boundary-anchored null geodesic $\gamma$,
if $\Omega$ is not constant on $\gamma$, then $\Omega \sim c/|\lambda|$ as either $\lambda
\rightarrow \infty$ or $\lambda \rightarrow -\infty$ for some $c>0$. Assume without
loss of generality that it vanishes as $\lambda \rightarrow \infty$.
In Poincar\'e coordinates of pure AdS, one can show that an affine null geodesic near the boundary has coordinate $z(\lambda) = \frac{ \alpha }{ |\lambda| }$ for some
$\alpha$ not depending on $\lambda$. The same scaling will hold asymptotically at large $\lambda$ in any
AlAdS spacetime, as can be seen by using a Fefferman-Graham expansion near $\partial M$ \cite{fefferman_1985}. Hence we have the asymptotic scaling $\Omega \sim
z$ as $\lambda \rightarrow \infty$ along $\gamma$.
In order for $(M, \Omega^2 g)$ to be AlAdS would need a function $\hat{z}$ such that $\hat{z}^2 \hat{g}=\hat{z}^2 \Omega^2 g$ extends smoothly to a metric
on $\partial M$. This requires the asymptotic scaling $\hat{z}\Omega \sim
z$ with some nonzero coefficient of proportionality or, equivalently, that $\hat{z}(\lambda)$ asymptotes to a nonzero constant as
$\lambda \rightarrow \infty$. However, this contradicts the requirement $\hat{z}|_{\partial M} = 0$
of a defining function, so $(M, \Omega^2 g)$ cannot be asymptotically AdS.
It follows that if $(M, \Omega^2 g)$ is AlAdS, then $\Omega$ is constant along every boundary-anchored null geodesic $\gamma$.
Since each $p\in A$ has a null geodesic ending on the boundary and $\Omega$ has a fixed value there, $\Omega$ is constant on $A$.
\end{proof}

\thmAF*
\begin{proof}

By definition of asymptotically flat \cite{Wald}, there exists a function $z$ on $M$ such that $z^2 g$ extends smoothly to a metric $\partial M = \mathscr{I}^+ \cup \mathscr{I}^-$, with $z|_{\partial M} = 0$, $\dd z|_{\partial M} \neq 0$ and $z|_{\interior M}>0$. 
Now, any asymptotically flat space has an asymptotic coordinate system near $\mathscr{I}^\pm$ given by \cite{Wald}
\begin{equation}
\dd s^2 = \frac{1}{z^2}(\dd u \, \dd z + g_{\mathbb{S}^{D-2}}) + \mathcal{O}(z^{-1}),
\end{equation}
where $g_{\mathbb{S}^{D-2}}$ is the metric of the unit $(D-2)$--sphere.
To leading order in $z$ and $\lambda$, the solution of the null geodesic equation for the $z$--coordinate is $z(\lambda) = c/|\lambda|$
where $c>0$ for a boundary-anchored geodesic.
Notice that $z$ is defined only up to a rescaling $z \rightarrow f(z) z$ by a strictly positive function $f \in C^2(M\cup \partial M)$. If $(M, \Omega^2 g)$ is asymptotically flat, then there is a function $\hat{z}$ satisfying all properties above in the spacetime $(M, \Omega^2 g)$. Since $z$ is uniquely defined up to any such $f$, so is $\hat{z}$, i.e., $\hat{z} = f(z) z/\Omega$.
Assume now $(M, \Omega^2 g)$ obeys the NCC, and that there exists a boundary-anchored null geodesic $\gamma $ on which $\Omega$ is somewhere non-constant. Then $\Omega\sim \mathcal{O}(1/|\lambda|)$ either as $\lambda \rightarrow \infty$ or $\lambda \rightarrow -\infty$. Assume without loss of generality the former. At large $\lambda$ we then have $\Omega \sim \mathcal{O}(1/|\lambda|) \sim \mathcal{O}(z)$ with a nonzero coefficient. This implies $\hat{z}|_{\partial M} \neq 0$, yielding a contradiction. Thus $(M, \Omega^2 g)$ is asymptotically flat only if $\Omega$ is constant along every boundary-anchored null geodesic. Since each $p\in A$ has a null geodesic ending on the boundary and $\Omega$ has a fixed value there, $\Omega$ is constant on $A$.
\end{proof}

\thmbtz*
\begin{proof}
Firstly, note the curious property that any spherically-symmetric BTZ black hole has circular photon orbits at every possible radius outside the horizon \cite{Cruz:1994ir}. In other words, there exist complete null geodesics at every fixed radius in both exteriors along which $\Omega$ must be constant by Lemma \ref{lem:cases}. This means an NCC-preserving Weyl transformation on BTZ can only possibly have radial dependence outside the horizon.
Finally, if $\lim_{r \rightarrow \infty}\omega(r)$ and $\lim_{r\rightarrow \infty}\omega'(r)$ exist, then \eqref{eq:omegabound} implies $\omega$ is monotonically non-decreasing along a radial null geodesic moving outwards, meaning that $\omega'(r)\ge0$ outside the horizon. 
\end{proof}

\section{Discussion}
\label{sec:discuss}

\paragraph{Relation to previous work:}
In this brief note we have shown how the combination of causal structure and the null curvature condition significantly constrain the conformal factor of the metric. To our knowledge, these results are new and have not appeared elsewhere in the literature.
There are nonetheless some recent results in \cite{Galloway:2021csq} of a similar flavor. In their work, it is shown that conformal rescalings of a class of generalized Robertson-Walker spacetimes must introduce geodesic incompleteness when requiring the strong energy condition. This fits well with the pattern explored in our paper, and shows that useful results can be obtained also for spacetimes with null curvature. 

We also note that the interaction between conformal transformations, geodesic completeness, and global causality conditions (without the consideration of energy conditions) were studied in \cite{Clarke:1970pi, Beem1976}.

\paragraph{Metric reconstruction in AdS/CFT with light-cone cuts:} In \cite{Engelhardt:2016wgb, Engelhardt:2016crc, Hernandez-Cuenca:2020ppu}, it was shown that knowledge of $(D+2)$-point correlators of local CFT operators on the conformal boundary $\partial M$ of a $D$-dimensional AlAdS spacetime is enough to reconstruct the conformal metric in a large portion of the causal wedge $C_W=J^+(\partial M) \cap J^-(\partial M)$. This method of bulk reconstruction requires no knowledge of the bulk equations of motion, allows for obtaining the internal dimensions of the bulk spacetime which trivialize asymptotically \cite{Hernandez-Cuenca:2020ppu}, and currently constitutes the only constructive approach to partially obtaining general bulk metrics from CFT data.\footnote{The reconstruction of bulk operators is a different problem which generally assumes knowledge of a background bulk geometry in the first place -- see \cite{Hamilton:2005ju} for more details. Of course, operator reconstruction includes perturbative metric reconstruction as a special case, but this still requires a reference background on top of which to consider graviton perturbations.}
The downside is that light-cone cuts recover only the causal structure of the spacetime, i.e., the conformal class of its metric. But if the bulk has vanishing null curvature, our results provide the missing conformal factor, since from Theorem \ref{thm:thmads} only a single choice of conformal factor will lead to a NCC-respecting AlAdS geometry. Note that our method only reconstructs the conformal factor on a subset of $C_W$ which is probed by boundary-anchored null geodesics. This constitutes a subset of the region where light-cone cuts reconstruct the conformal metric, but as long as $D\geq 4$ we expect both regions to be large subsets of $C_W$. For $D=3$ this is not necessarily true due to the fact that gravitational interactions do not fall off with distance. In the case of the BTZ black hole this causes all non-spherical null geodesics to be incomplete, leading to the failure of light-cone cuts. Our methods fail at determining completely the conformal factor anywhere in $C_W$ in the BTZ geometry for the same reason.

\paragraph{Future directions:} Instead of analyzing complete null geodesics, one could study the effect of Weyl transformations on either complete timelike geodesics or complete spatial slices. In these cases the more useful energy condition would likely be the weak energy condition.\footnote{We could also consider the dominant or strong energy conditions, but the latter is not physically well motivated and the former is stronger than the weak one.} Carrying out the analysis on a spatial slice would require PDE analysis, but could on the other hand have the potential to constrain the conformal factor also in black hole interiors. 

In cases where the bulk has null curvature, we would get constraints on the conformal factor through~\eqref{eq:ubound} (and the less constraining but more explicit~\eqref{eq:ancc2}). It would be interesting to study more quantitatively how much \eqref{eq:ubound} constrains geometry, and thus how good the approximate bulk reconstruction is when null curvature is present.

The geodesic NCC in \eqref{eq:secschw} can be recast into the bound $S(f_\gamma) \leq \frac{2}{D-2}R_{kk}$, where $S$ is the Schwarzian derivative and $f_\gamma'(\lambda) = e^{2\omega(\lambda)}$. It would be interesting to investigate if the well-studied properties of the Schwarzian derivative would allow us to derive rigidity results when $R_{kk}\neq 0$. It could also be useful to try to study the NCC directly at the level of \eqref{eq:expncc}, i.e., more globally rather than restricting to single null geodesics as we did. A possible step forward in this direction is to re-interpret \eqref{eq:expncc} as the NEC for matter fields in a Brans-Dicke theory with $\omega$ playing the role of the dilatonic scalar \cite{Brans:1961sx}. This way, the NCC could be studied as the physicality condition that it be possible to add matter fields to Brans-Dicke such that solutions to the equations of motion obey the NEC.

Finally, in Theorem \ref{thm:btz}, we explored the implications of our results on BTZ, deriving a condition on the conformal factor for spacetimes in its conformal class to obey the NCC. This low-dimensional example illustrated the usefulness of our methods even in the absence of boundary-anchored null geodesics.
Higher-dimensional spacetimes provide an arena where Theorems \ref{thm:thmads} and \ref{thm:thmaf} would play an important role in large regions of spacetime, but where we also expect there to be regions in $C_W$ which are not probed by boundary-anchored null geodesics, particularly near black holes. Hence it would be interesting to explore the implications of our methods in these regions of higher-dimensional black hole spacetimes \cite{Emparan:2008eg,Horowitz:2012nnc}.

\acknowledgments
It is a pleasure to thank Netta Engelhardt for collaboration in the early stages of this project, Gary Horowitz for helpful comments on various parts of the paper, and Gregory Galloway and Elena Giorgi for guidance on the existing mathematical literature.
\AA{}F is supported in part by NSF grant PHY-2011905 and the MIT department of physics, and in part by an Aker Scholarship.
SHC is supported by NSF grant PHY-1801805 and the University of California, Santa Barbara.


\addcontentsline{toc}{section}{References}
\bibliographystyle{JHEP}
\bibliography{references.bib}

\end{document}